%% file: 2023-crm-ultr.tex
\documentclass[sigconf,natbib=true,anonymous=false,review=false,screen=true]{acmart}

\acmSubmissionID{6205}

\input{packages}

\input{copyright}

\input{definitions}

\input{authors}

\title[Safe Deployment for Counterfactual Learning to Rank with Exposure-Based Risk Minimization]{Safe Deployment for Counterfactual Learning to Rank\\ with Exposure-Based Risk Minimization}

\allowdisplaybreaks

\begin{document}

\begin{abstract}
\Ac{CLTR} relies on exposure-based \ac{IPS}, a \acs{LTR}-specific adaptation of \ac{IPS} to correct for position bias.
While \ac{IPS} can provide unbiased and consistent estimates, it often suffers from high variance.
Especially when little click data is available, this variance can cause \ac{CLTR} to learn sub-optimal ranking behavior.
Consequently, existing \ac{CLTR} methods bring significant risks with them, as naively deploying their models can result in very negative user experiences. 

We introduce a novel risk-aware \ac{CLTR} method with theoretical guarantees for safe deployment.
We apply a novel exposure-based concept of risk regularization to \ac{IPS} estimation for \acs{LTR}.
Our risk regularization penalizes the mismatch between the ranking behavior of a learned model and a given safe model.
Thereby, it ensures that learned ranking models stay close to a trusted model, when there is high uncertainty in \ac{IPS} estimation,
which greatly reduces the risks during deployment.  
Our experimental results demonstrate the efficacy of our proposed method, which is effective at avoiding initial periods of bad performance when little date is available, while also maintaining high performance at convergence.
For the \ac{CLTR} field, our novel exposure-based risk minimization method enables practitioners to adopt \acs{CLTR} methods in a safer manner that mitigates many of the risks attached to previous methods.
\end{abstract}




\maketitle

\acresetall

\input{sections/01-introduction}

\input{sections/02-related-works}

\input{sections/03-background}

\input{sections/04-method}

\input{sections/05-experiments}

\input{sections/06-results}

\input{sections/07-conclusion}

\section*{Acknowledgements}
This research was supported by Huawei Finland and by the Hybrid Intelligence Center, a 10-year program funded by the Dutch Ministry of Education, Culture and Science through the Netherlands Organisation for Scientific Research, \url{https://hybrid-intelligence-centre.nl}. 
This work used the Dutch national e-infrastructure with the support of the SURF Cooperative using grant no. EINF-4963.
All content represents the opinion of the authors, which is not necessarily shared or endorsed by their respective employers and/or sponsors.

\section*{Reproducibility}
All experimental results in this work were obtained using publicly available data.
Our implementation is publicly available at \url{https://github.com/shashankg7/crm_ultr}.

\balance
\bibliographystyle{ACM-Reference-Format}
\bibliography{bibliography}

\end{document}

%% file: packages.tex
\usepackage[inline]{enumitem}
\usepackage{acronym}
\usepackage{xkcdcolors}

\PassOptionsToPackage{dvipsnames}{xcolor}

\usepackage{dsfont}
\usepackage{natbib}
\usepackage[noend]{algorithmic}
\usepackage{algorithm}
\usepackage{bbm}
\usepackage{beramono}
\usepackage{bm}
\usepackage{booktabs}
\usepackage[skip=0pt]{caption}
\usepackage{cleveref}
\usepackage[T1]{fontenc}
\usepackage{graphicx}
\usepackage{hyperref}
\usepackage{listings}
\usepackage{multirow}
\usepackage{natbib}
\usepackage{subcaption}
\usepackage{tabularx}
\usepackage[dvipsnames]{xcolor}
\usepackage{numprint}
\usepackage[normalem]{ulem}
\usepackage{multirow}
\usepackage{float}
\usepackage{multicol, blindtext}
\usepackage{amsthm}
\usepackage[english]{babel}
\usepackage{amsthm}
\usepackage[T1]{fontenc}
\usepackage{mleftright}
\usepackage{tikz}
\usepackage{color}
\usepackage{graphicx}
\usetikzlibrary{backgrounds}
\usepackage{subcaption}
\usepackage{adjustbox}
\usepackage{mathtools}
\usepackage{appendix}
\usepackage{enumitem}

%% file: copyright.tex
\copyrightyear{2023}
\acmYear{2023}
\setcopyright{rightsretained}
\acmConference[SIGIR '23]{Proceedings of the 46th International ACM SIGIR Conference on Research and Development in Information Retrieval}{July 23--27, 2023}{Taipei, Taiwan}
\acmBooktitle{Proceedings of the 46th International ACM SIGIR Conference on Research and Development in Information Retrieval (SIGIR '23), July 23--27, 2023, Taipei, Taiwan}\acmDOI{10.1145/3539618.3591760}
\acmISBN{978-1-4503-9408-6/23/07}

%% file: definitions.tex
\acrodef{CF}{collaborative filtering}
\acrodef{LTR}{learning to rank}
\acrodef{NDCG}{normalized discounted cumulative gain}
\acrodef{DCG}{discounted cumulative gain}
\acrodef{VAE}{variational autoencoder}
\acrodef{VAE}{variational autoencoder}
\acrodef{ELBO}{evidence lower bound objective}
\acrodef{IPS}{inverse propensity scoring}
\acrodef{BPR}{bayesian personalized ranking}
\acrodef{MF}{matrix factorization}
\acrodef{MNAR}{missing-not-at-random}
\acrodef{ULTR}{unbiased learning-to-rank}
\acrodef{CLTR}{counterfactual learning to rank}
\acrodef{LOLN}{law of large numbers}
\acrodef{CRM}{counterfactual risk minimization}
\acrodef{IS}{importance sampling}
\acrodef{i.i.d}{independent and identically distributed}
\acrodef{CRM}{counterfactual risk minimization}
\acrodef{PL}{Plackett-Luce}
\acrodef{CTR}{click through rate}
\acrodef{SEA}{safe exploration algorithm}
\acrodef{GENSPEC}{generalization and specialization }
\acrodef{VCRM}{variational counterfactual risk minimization}
\acrodef{SGD}{stochastic gradient descent}

\theoremstyle{definition}

\newcommand{\headernodot}[1]{\vspace{1mm}\noindent\textbf{#1}}
\newcommand{\header}[1]{\headernodot{#1.}}

\allowdisplaybreaks

%% file: authors.tex
\author{Shashank Gupta}
\orcid{0000-0003-1291-7951}
\affiliation{%
  \institution{University of Amsterdam}
  \city{Amsterdam}
  \country{The Netherlands}
}
\email{s.gupta2@uva.nl}

\author{Harrie Oosterhuis}
\orcid{0000-0002-0458-9233}
\affiliation{%
  \institution{Radboud Universiteit}
  \city{Nijmegen}
  \country{The Netherlands}
}
\email{harrie.oosterhuis@ru.nl}

\author{Maarten de Rijke}
\orcid{0000-0002-1086-0202}
\affiliation{%
  \institution{University of Amsterdam}
  \city{Amsterdam}
  \country{The Netherlands}
}
\email{m.derijke@uva.nl}

%% file: sections/01-introduction.tex
\section{Introduction}
\label{sec:intro}
\Ac{LTR} methods optimize ranking systems so that the resulting ranking behavior maximizes a given ranking metric~\cite{liu2009learning}. 
Traditionally, most \ac{LTR} methods applied a supervised learning procedure based on manually-created relevance judgements.
However, obtaining such judgements is time-consuming, expensive and does not scale~\citep{chapelle2011yahoo, qin2010letor}. 
As an alternative, \ac{LTR} methods have been developed that rely on clicks, as they are much cheaper to obtain in abundance in the form of user interaction logs~\citep{joachims2002optimizing}.

Despite its low costs, click data is generally strongly affected by different forms of interaction bias.
Interactions with rankings often suffer from \textit{position bias}~\cite{craswell2008experimental}: the position at which an item was shown often affects its \ac{CTR} more than its relevance.
As a result, the clicks observed in interaction logs are often more reflective of where items were displayed during logging than how relevant users find them.
Thus, naively using this data for \ac{LTR}, without corrections, can result in heavily \emph{biased} models with suboptimal ranking performance~\citep{joachims2017unbiased, wang2016learning}.

To mitigate the bias problem in interaction data, the field of \ac{CLTR} has proposed methods to mitigate bias with unbiased estimation~\citep{joachims2017unbiased}.
\Ac{CLTR} mainly relies on exposure-based \ac{IPS}~\cite{oosterhuis2020unbiased,wang2018position}, a \ac{LTR} specific adaptation of the \ac{IPS} counterfactual estimation method~\cite{swaminathan2015batch, joachims2016counterfactual, horvitz1952generalization}.
Standard exposure-\ac{IPS} weights clicks by the inverse effect of position-bias on the clicked item.
This procedure thus gives more weight to clicks on items that are underrepresented due to position-bias, and vice versa.
In expectation, this removes the effect of position-bias from the loss that is optimized.
 
\header{Unsafe \ac{CLTR}}
Despite enabling unbiased optimization, \ac{IPS} is also known to suffer from high variance~\cite{joachims2017unbiased, oosterhuis2022doubly}. 
In cases with a lack of click data or with large amounts of noise, high variance can make IPS-based \ac{CLTR} unreliable and lead to very sub-optimal ranking models~\citep{jagerman2020safe, oosterhuis2021robust}.
This problem can be so severe that the learned ranking models can be worse than the model used to log the interaction data.
Deploying such a learned model could result in a substantially degraded user experience. 
In other words, despite the improvements that IPS-based \ac{CLTR} can bring, it is also an \emph{unsafe} approach since it may lead to considerable deteriorations.
%, under certain circumstances.

This (un)safety issue is not unique to \ac{IPS}-based \ac{CLTR}. 
\citet{swaminathan2015batch} address this issue for contextual bandit problems by applying a generalization bound.
Such a bound can provide a high-confidence upper limit on the difference between the true and estimated performance of a bandit policy~\cite{shalev2014understanding, thomas2015high}.
This allows for safer \emph{conservative} optimization.
For instance, \citet{wu2018variance} introduce a bound based on the divergence between the new policy and the logging policy.
This bound avoids policies that stray away from the logging policy, unless there is strong evidence that they are actual improvements.
This method might appear to be a great fit for \ac{CLTR}, but, unfortunately, it is based on action propensities that do not generalize well to the very large action spaces in \ac{CLTR}.
Therefore, there is a need for a conservative generalization bound that is practical and effective in the \ac{CLTR} setting.

\header{Safe \ac{CLTR}}
To address this gap, we propose an exposure-based \ac{CRM} method that is specifically designed for safe \ac{CLTR}.
Similar to how exposure-based \ac{IPS} deals with the large action spaces in ranking settings, our method is based on an exposure-based alternative to action-based generalization bounds.
We first introduce a divergence measure based on differences between the distributions of exposure of a new policy and a safe logging policy.
Then we provide a novel generalization bound and prove that it is a high-confidence lower-bound on the performance of a learned policy.
When uncertain, this bound defaults to preferring the logging policy and thus avoids decreases in performance due to variance.
In other words, with high-confidence, ranking models optimized with this bound are guaranteed to never deteriorate the user experience, even when little data is available.

\header{Main contributions}
We are the first to address \ac{CRM} for \ac{CLTR} and contribute a novel exposure-based \ac{CRM} method for safe \ac{CLTR}.
Our experimental results show that our proposed method is effective at avoiding initial periods of bad performance when little date is available, while also maintaining high performance at convergence.
Our novel exposure-based \ac{CRM} method thus enables safe \ac{CLTR} that can mitigate many of risks attached to previous methods.
We hope that our contribution makes the adoption of \ac{CLTR} methods more attractive to practitioners working on real-world search and recommendation systems.

%% file: sections/02-related-works.tex
\section{Related Work}
We review related work on \ac{CLTR} and \ac{CRM} in off-policy learning. 

\subsection{Counterfactual learning to rank}
\Ac{LTR} deals with learning optimal rankings to maximize a pre-defined notion of utility~\citep{liu2009learning}.
Traditionally, \ac{LTR} systems were optimized using supervised learning on manually-created relevance judgements~\citep{chapelle2011yahoo}.
But the manual curation of relevance judgements is a time-consuming and costly process~\citep{chapelle2011yahoo,qin2010letor}.
Also, manually-graded relevance signals do not always align well with actual user preferences~\citep{sanderson2010user}. 
Due to these shortcomings, \ac{LTR} from user interactions has become a popular alternative to supervised \ac{LTR}~\citep{speretta2005personalized,jiang2016learning,chapelle2009dynamic,joachims2017unbiased}.

Learning from user interactions/click logs was introduced in the pioneering work of~\citet*{joachims2002optimizing}.
Click data is relatively cheap to collect and indicative of actual user preferences~\cite{radlinski2008does}.
In spite of these advantages, click data is known to be a noisy and biased estimate of the true user preferences~\cite{craswell2008experimental,oosterhuis2020unbiased}.
Some of the common biases identified in the \ac{LTR} literature are position bias~\cite{craswell2008experimental}:
trust bias~\cite{agarwal2019addressing}, and item-selection bias~\cite{oosterhuis2020policy}.
 
To counter the effect of bias, \citet{joachims2017unbiased} introduced counterfactual learning in the context of \ac{LTR}. 
They proposed the application of \acf{IPS}, a causal inference technique that has prevalence in the offline bandit learning literature~\citep{joachims2016counterfactual}.
\ac{IPS} models the probability of the user examining a document at a given displayed rank.
The inverse of the examination probability, i.e., the inverse propensity, is used to correct for the position bias. 
    As a result of the inverse weighing scheme, \ac{IPS}-based \ac{LTR} optimization is unaffected by position bias, in expectation~\citep{joachims2017unbiased}.
Since its introduction, there has been an increasing interest in the area, with several application of \ac{IPS} in the context of ranking~\citep{agarwal2019addressing,vardasbi2020inverse,oosterhuis2020policy,wang2018position}.
Recent work has also explored \ac{CLTR} under a stochastic logging policy, where some exploration is introduced, as opposed to pure exploitation~\citep{oosterhuis2021unifying,oosterhuis2020policy,yadav2021policy}.

With regard to safety in learning from user interactions, \citet{jagerman2020safe} introduced the notation of safe exploration for offline contextual bandit algorithms. 
The authors introduced \ac{SEA}, which applies high-confidence performance bounds to \emph{safely} choose between the deployment of a logging policy and a learned policy.
\citet{oosterhuis2021robust} applied this context to \ac{LTR} and introduced a generalization and specialization framework to safely choose between a generalized feature-based \ac{LTR} model, and a specialized tabular \ac{LTR} model.
The important difference between prior work and our work is that existing methods safely \emph{choose} between policies, whereas our method safely \emph{optimizes} a policy.
To the best of our knowledge, we are the first to consider notion of safety for the \emph{optimization} of \ac{LTR} models.

\subsection{Counterfactual risk minimization for offline learning from logs}
A relevant area closely related to \ac{CLTR} is off-policy learning, or offline learning from bandit feedback data~\citep{joachims2016counterfactual,saito2021counterfactual,swaminathan2015batch,he2019off}.
Off-policy learning tries to bridge the mismatch between the action distributions of a new policy and the logging policy~\cite{joachims2016counterfactual}.
The most common techniques used to achieve that goal are \ac{IPS} and importance sampling~\cite{horvitz1952generalization}.
However, as noted by~\citet{cortes2010learning}, the \ac{IPS} estimator can have unbounded variance, which can lead to large errors in its estimation.
Consequently, optimization with IPS can result in convergence problems and severely suboptimal policies. 

To account for this high-variance problem, \citet{swaminathan2015batch} introduced \acf{CRM}, an off-policy method that explicitly controls for the variance during off-policy learning from bandit feedback data. 
Their learning objective consists of both the \ac{IPS} loss and a variance regularization term, which minimizes the dissimilarity between the two policies.
This variance regularization term represents the \emph{risk} that stems from the variance of the IPS estimation.
Computing it requires a pass over the entire data which does not scale well. 
As a scalable alternative, \citet{wu2018variance} introduced \ac{VCRM}, where the authors estimate the \emph{risk} of the new policy by random sampling from the logged data. 
The objective function to be optimized in the \ac{VCRM} method is derived from a generic theoretical analysis of learning from importance sampling~\cite{cortes2010learning}. 
The risk term in the \ac{VCRM} method is defined in terms of a specific divergence between the logging policy and the new policy, known as the R\'enyi divergence~\citep{renyi1961measures}. 
To the best of our knowledge, there is no work on \ac{CRM} in a \ac{LTR} setting, making our work the first to propose a \ac{CRM} approach for the \ac{LTR} task.

%% file: sections/03-background.tex
\section{Background}

\subsection{Learning to rank}
The objective of learning to rank methods is to optimize a ranking policy ($\pi$), so that for user-issued queries ($q$) it provides the optimal ranking of their pre-selected candidate document sets ($D_q$)~\citep{liu2009learning}.
Formally, this objective can be expressed as the maximization of the following utility function:
\begin{equation}
    U(\pi) =  \mathbb{E}_{q} \mleft[ \sum_{d \in D_q} \rho(d \mid q, \pi) P(R=1 \mid d, q) \mright]. \label{true-utility}
\end{equation}
where $\rho(d \mid q, \pi)$ is the weight $\pi$ gives to document $d$ for query $q$.
The choice of $\rho$ determines what metric is optimized, for instance, the well-known \ac{NDCG} metric~\citep{jarvelin2002cumulated}:
\begin{equation}
    \rho_{\text{DCG}}(d \mid q, \pi) = \mathbb{E}_{y \sim \pi( \cdot \mid q)} \mleft[ (\log_2(\textrm{rank}(d \mid y) + 1))^{-1} \mright]. \label{rho}
\end{equation}
where $y$ is a ranking sampled from the policy $\pi$. For this paper, the aim is to optimize the expected number of clicks, the next subsection will explain how we choose $\rho$ accordingly.

\subsection{Counterfactual learning to rank}
\label{sec:background:cltr}

\textbf{Position bias in clicks.}
Optimizing the \ac{LTR} objective in Eq.~\ref{true-utility} requires access to the true relevance labels ($P(R=1 \mid d, q)$), which is often impossible in real-world ranking settings. 
As an alternative, \ac{CLTR} uses clicks, since they are present in abundance as logged user interactions. 
However, clicks are a biased indicator of relevance; for this paper, we will assume the relation between clicks and relevance is determined by a position-based click model~\citep{joachims2017unbiased, chuklin-click-2015}.
For a document $d$ displayed in ranking $y$ for query $q$, this means the click probability can be decomposed into a rank-based examination probability and a document-based relevance probability:
\begin{equation}
    P(C=1 \mid d, q, y) = P(E=1 \mid \text{rank}(d \mid y))  P(R=1 \mid d, q).   
    \label{click-model}
\end{equation}
The key characteristic of the position-based click model is that the probability of examination only depends on the rank at which a document is displayed: $P(E=1 \mid d, q, y) = P(E=1 \mid \text{rank}(d \mid y))$.
Furthermore, this model assumes that clicks only take place when a document is both relevant to a user and examined by them.
Consequently, the click signal is an indication of both the relevance and examination of documents.
Thus, the position at which a document is displayed can have a stronger effect on its click probability than its actual relevance~\citep{craswell2008experimental}.

\header{Inverse-propensity-scoring for \ac{CLTR}}
We assume a setting where $N$ interactions have been logged using the logging policy $\pi_0$, for each interaction $i$ the query $q_i$, the displayed ranking $y_i$, and the clicks $c_i$ are logged:
\begin{equation}
    \mathcal{D} = \big\{q_i, y_i, c_i \big\}^N_{i=1}. \label{logs}
\end{equation}
We will use $c_i(d) \in \{0,1\}$ to denote whether document $d$ was clicked at interaction $i$.
Furthermore, we choose $\rho$ to match the examination probabilities under $\pi$:
\begin{equation}
    \rho(d \mid q, \pi) =  \mathbb{E}_{y \sim \pi( \cdot \mid q)} \big[ P(E=1 \mid \text{rank}(d \mid y))  \big] = \rho(d).
    \label{eq:exposurenew}
\end{equation}
Hence, our optimization objective $U(\pi)$ is equal to the expected number of clicks (cf.\ Eq.~\ref{true-utility} and~\ref{click-model}).

In order to apply \ac{IPS}, we need the propensity of each document~\citep{joachims2017unbiased}, following \citet{oosterhuis2021unifying} we use:
\begin{equation}
\begin{split}
    \rho(d \mid q, \pi_0) &=  P(E =1  \mid \pi_{0}, d, q)\\
        &=  \mathbb{E}_{y \sim \pi_{0}( \cdot \mid q)} \big[ P(E=1 \mid \text{rank}(d \mid y))  \big] = \rho_0(d).
\label{policy-aware-exposure}
\end{split}
\end{equation}
Thus, the exposure of $d$ represents how likely it is examined when using $\pi_0$ for logging.
Thereby, it indicates how much the clicks on $d$ underrepresent its relevance.
For the sake of brevity, we drop $q$, $\pi$ and $\pi_0$ from our notation when their values are clear from the context: i.e., $\rho(d \mid q, \pi) = \rho(d)$ and $\rho(d \mid q, \pi_0) = \rho_0(d)$.

The exposure-based \ac{IPS} estimator takes each click in $\mathcal{D}$ and weights it inversely to $\rho_{0}(d)$ to correct for position-bias~\citep{joachims2017unbiased, oosterhuis2021unifying}:
\begin{equation}
    \hat{U}(\pi) = \frac{1}{N} \sum_{i=1}^{N} \sum_{d \in D_{q_i}}  \frac{\rho(d)}{\rho_{0}(d)} c_i(d).
    \label{cltr-obj}
\end{equation}
In other words, to compensate that position bias lowers the click probability a document by a factor of $\rho_{0}(d)$, clicks are weighted by $1/\rho_{0}(d)$ to correct for this effect in expectation.
As a result, clicks on documents that $\pi_0$ is likely to show at positions with low examination probabilities (i.e., the bottom of a ranking) receive a higher \ac{IPS} weight to compensate.

\header{Statistical properties of the IPS estimator}
The \ac{IPS} estimator $\hat{U}(\pi)$ (Eq.~\ref{cltr-obj}) is an unbiased and consistent estimate of our \ac{LTR} objective $U(\pi)$ (Eq.~\ref{true-utility})~\cite{oosterhuis2022reaching}. 
It is \emph{unbiased} since its expected value is equal to our objective:
\begin{equation}
\mathbb{E}_{q,y,c} \mleft[ \hat{U}(\pi) \mright] = U(\pi),
\end{equation}
and it is \emph{consistent} because this equivalence also holds in the limit of infinite data:
\begin{equation}
\lim\limits_{N \to \infty} \hat{U}(\pi) = U(\pi).
\end{equation}
For proofs of these properties, we refer to previous work~\cite{oosterhuis2020policy,joachims2017unbiased,oosterhuis2020learning}.

Importantly, the unbiasedness and consistency properties do not indicate that the actual IPS estimates will be reliable.
This is because the estimates produced by IPS are also affected by its variance:
\begin{equation}
    \mathrm{Var}_{y,c}\mleft[\hat{U}(\pi) \mid q\mright] = \sum_{d \in D_q}  \frac{\rho(d)^2}{\rho_{0}(d)^2} \mathrm{Var}_{y,c}\mleft[ c(d) \mid \pi_0, q\mright].
\end{equation}
The variance is large when some propensities are small, due to the $\rho_{0}(d)^{-2}$ term.
Hence, the actual estimates that IPS produces may contain large errors, especially when $N$ is relatively small or clicks are very noisy.
Thus, $\hat{U}(\pi)$ may be far removed from the true $U(\pi)$, and optimization with IPS may be unsafe and lead to unpredictable results.

\begin{figure*}[t]
\includegraphics[width=\linewidth,trim= 0 -1cm 0 0,clip]{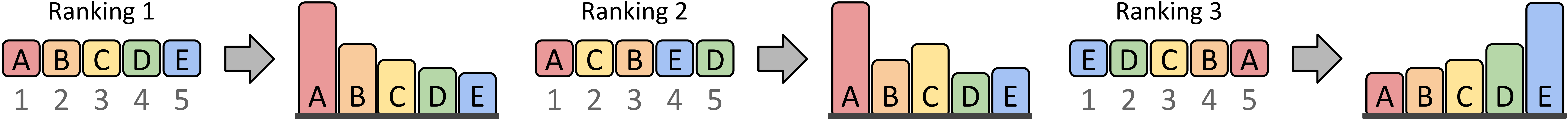}
\caption{
Three rankings and their normalized expected exposure distributions (Eq.~\ref{norm-exposure}) based on DCG weights (Eq.~\ref{rho}).
According to our exposure-based divergence, ranking 1 and ranking 2 are quite similar despite only agreeing on the placing of document A.
In contrast, ranking 1 and ranking 3 also agree on the placement of a single document (C) but have the highest possible dissimilarity, due to their highly mismatched exposure distributions.
}
\label{fig:exposure}
\end{figure*}

\subsection{Counterfactual risk minimization for offline bandit learning}\label{crm-bandit}

The foundational work by~\citet{swaminathan2015batch}  introduced the idea of \acf{CRM} for off-policy learning in a contextual bandit setup.
To avoid the negative effects of high-variance with \ac{IPS} estimation during bandit optimization, they utilize a generalization bound through the addition of a risk term~\citep{maurer2009empirical}.
With a probability of $1-\delta$, the \ac{IPS} estimate minus the risk term is a lower bound on the true utility of the policy: 
\begin{equation}
    P\big( U(\pi) \geq \hat{U}(\pi) - \text{Risk}(\delta) \big) > 1-\delta.
    \label{eq:generalizationbound}
\end{equation}
Therefore, optimization of the lower bound can be more reliable than solely optimizing the \ac{IPS} estimate ($\hat{U}(\pi)$), since it provides a high-confidence guarantee that a lower bound on the \emph{true} utility of the policy is maximized.
\citet{swaminathan2015batch} propose using the sample variance as the risk factor:
\begin{equation}
    \hat{U}_\text{action-CRM}(\pi) = \hat{U}_\text{action}(\pi) - \lambda \sqrt{\frac{1}{N} \mathrm{Var}\big[\hat{U}_\text{action}(\pi)\big]}, \label{crm-loss}
\end{equation}
where $\lambda \in \mathbb{R}^{>0}$ is an alternative to the $\delta$ parameter that also determines how probable it provides a bound on the true utility.
Importantly, this bound is based on an action-based \ac{IPS} estimator. 
For our \ac{LTR} setting this would translate to:
\begin{equation}
\hat{U}_\text{action}(\pi)
=
\frac{1}{N}\sum^{N}_{i=1} \frac{\pi(y_i \mid q_i)}{\pi_0(y_i \mid q_i)}
\sum_{d \in D_{q_i}} c_i(d).
\label{eq:actionips}
\end{equation}
Action-based \ac{IPS} estimation does not work well in the \ac{LTR} setting because the large number of possible rankings results in extremely small action propensities: $\pi_0(y_i \mid q_i)$, creating a high-variance problem.
As discussed in Section~\ref{sec:background:cltr}, for this reason \ac{CLTR} uses exposure-based propensities instead (Eq.~\ref{policy-aware-exposure} and~\ref{cltr-obj}), as they effectively avoid extremely small values.
As a result, the \ac{CRM} approach from \citep{swaminathan2015batch} is not effective for \ac{CLTR}, since the high-variance of its action-based \ac{IPS} make the method impractical in the ranking setting.

Another downside of the \ac{CRM} approach is that the computation of the sample-variance requires a full-pass over the training dataset, which is computationally costly for large-scale datasets. 
As a solution, \citet{wu2018variance} introduce variational \ac{CRM} (VCRM) which uses an upper bound on the variance term based on the R\'enyi divergence between the new policy and the logging policy~\citep{renyi1961measures}.
This R\'enyi divergence is approximated via random sampling, thus making the VCRM method suitable for stochastic gradient descent-based training methods~\cite{nowozin2016f}. 
Nevertheless, this \ac{CRM} approach still relies on action-based propensities, and therefore, does not provide an effective solution for the high-variance problem in \ac{CLTR}.

%% file: sections/04-method.tex
\section{A Novel Exposure-Based Generalization Bound for CLTR}

To develop a \ac{CRM} method for \ac{CLTR} with safety gaurantees, we aim to find a risk term that gives us a generalization bound as in Eq.~\ref{eq:generalizationbound}.
Importantly, this bound has to be effective in the \ac{LTR} setting, therefore, our approach should avoid action-based propensities. 
We take inspiration from work by \citet{wu2018variance}, who use the fact that the R\'enyi divergence is an upper bound on the variance of an \ac{IPS} estimator:
\begin{equation}
    \mathrm{Var} \mleft[ \hat{U}_\text{action}(\pi) \mright] \leq d_2(\pi \,\Vert\, \pi_0),
    \label{divergence-bandit}
\end{equation}
where $d_2$ is the exponentiated R\'enyi divergence between the new policy and the logging policy~\citep{renyi1961measures}:
 \begin{equation}
    d_2(\pi \,\Vert\, \pi_0) = \mathbb{E}_{q}\mleft[ \sum_{y} \mleft(\frac{\pi(y \mid q)}{\pi_0(y \mid q)}\mright)^2 \pi_0(y \mid q)  \mright].
    \label{eq:actionbaseddiv}
\end{equation}
In other words, the dissimilarity between the logging policy and a new policy can be used to bound the variance of the IPS estimate of the new policy's performance.
However, because this divergence is based on action propensities, it is not effective in the \ac{LTR} setting.
 
Below, we introduce an exposure-based measure of divergence that can produce a desired generalization bound for \ac{LTR} optimization.
Section~\ref{sec:normexposure} introduces the concept of normalized exposure that treats rankings as exposure distributions.
Section~\ref{sec:varbound} proves that R\'enyi divergence based on normalized exposure can bound the variance of an exposure-based \ac{IPS} estimator.
Section~\ref{sec:perfbound} uses this variance bound to construct a generalization bound for \ac{CLTR}.

\subsection{Normalized expected exposure}
\label{sec:normexposure}

R\'enyi divergence is only valid for probability distributions, e.g., $d_2(\pi \,\Vert\, \pi_0)$ with $\pi(y \mid q)$ and $\pi_0(y \mid q)$.
However, expected exposure is not a probability distribution, i.e., the values of $\rho(d)$ (Eq.~\ref{eq:exposurenew}) or $\rho_0(d)$ (Eq.~\ref{policy-aware-exposure}) do not necessarily sum up to one, over all documents to be ranked.
This is because users generally examine more than a single item in a single displayed ranking~\citep{craswell2008experimental}, as a result, expected exposure can be seen as a distribution of multiple examinations.
Our insight is that a valid probability distribution can be obtained by normalizing the expected exposure:
\begin{equation}
        \rho'(d) = \frac{ \rho(d)}{\sum_{d' \in D} \rho(d')}
        = \frac{\rho(d)}{\mathrm{Z}},
    \label{norm-exposure}
\end{equation}
where the normalization factor is a constant that only depends on $K$, the (truncated) ranking length:
\begin{equation}
\begin{split}
    \mathrm{Z} = \sum_{d \in D}^{} \rho(d) & = \sum_{d \in D} \mathbb{E}_{y \sim \pi } \big[ P\big(E=1 \mid \text{rank}(d \mid y)\big)  \big]   \\
    & = \mathbb{E}_{y \sim \pi} \Big[  \sum_{d \in D}  P\big(E=1 \mid \text{rank}(d \mid y)\big)  \Big]    \\
    & = \mathbb{E}_{y \sim \pi} \Big[  \sum_{k=1}^K  P\big(E=1 \mid k \big)  \Big]   =  \sum_{k=1}^K  P\big(E=1 \mid k \big).  
    \end{split} 
    \label{total-exposure}
\end{equation} 
In this way, $\mathrm{Z}$ can be seen as the expected amount of examination that any ranking will receive, and $\rho'$ as the probability distribution that indicates how it is expected to spread over documents.

An important property is that the ratio between two propensities is always equal to the ratio between their normalized counterparts:
\begin{equation}
    \frac{\rho(d)}{\rho_0(d)} = \frac{\rho'(d)}{\rho'_{0}(d)}.
    \label{eq:ratio}
\end{equation}
This is relevant to \ac{IPS} estimation since it only requires the ratios between propensities, the proofs in the remainder of this paper make use of this property.

Finally, using the normalized expected exposure, we can introduce the exponentiated exposure-based R\'enyi divergence:
\begin{align}
    d_2(\rho \,\Vert\, \rho_0) &=  \mathbb{E}_{q} \bigg[ \sum_{d \in D_{q}}^{} \rho'_{0}(d) \mleft( \frac{\rho'(d)}{\rho'_{0}(d)} \mright)^{2}  \bigg].
    \label{eq:exposuredivergence}
\end{align}
The key difference between our exposure-based divergence and action-based divergence is that it allows policies to be very similar, even when they have no overlap in the rankings  they produce.
As an intuitive example, Figure~\ref{fig:exposure} displays three different rankings and their associated normalized expected exposure distributions; these are the distributions for deterministic policies that give 100\% probability to one of the rankings.
Under action-based divergence, these policies would have the highest possible dissimilarity since they have no overlap in their possible actions, i.e., the rankings they give non-zero probability.
In contrast, exposure-based divergence gives high similarity between ranking 1 and ranking 2, since the differences in their exposure distribution are minor.
We note that these rankings still disagree on the placement of all documents except one.
Conversely, for ranking 1 and ranking 3, which also only agree on a single document placement, exposure-based divergence gives the lowest possible similarity score because their exposure distributions are highly mismatched.
Importantly, by solely considering differences in exposure distributions, exposure-based divergence naturally weighs differences at the bottom of rankings as less impactful than changes that affect the top.
As a result, exposure-based divergence more closely corresponds with common ranking metrics (Eq.~\ref{true-utility}) than existing action-based divergences.

\subsection{Exposure-divergence bound on variance}
\label{sec:varbound}

We now provide proof that exposure-based divergence is an upper bound on the variance of \ac{IPS} estimators for \ac{CLTR}.
\begin{theorem}
\label{var-theorm}
    Given a ranking policy $\pi$ and logging policy $\pi_{0}$, with the expected exposures $\rho(d)$ and $\rho_{0}(d)$ respectively, the variance of the exposure-based \ac{IPS} estimate $\hat{U}(\pi)$ is upper-bounded by exposure-based divergence:
    \begin{equation}
        \begin{aligned}
            \mathrm{Var}_{q,y,c}\mleft[\hat{U}(\pi)\mright] \leq  \frac{\mathrm{Z} }{N} d_2(\rho \,\Vert\, \rho_0).  \label{variance-full}
    \end{aligned}
    \end{equation} 
\end{theorem}
\begin{proof}
    From the definition of $\hat{U}(\pi)$ (Eq.~\ref{cltr-obj}) and the assumption that queries $q$ are \ac{i.i.d}, the variance of the counterfactual estimator can be rewritten as an expectation over queries~\citep{oosterhuis2020taking}:
    \begin{equation}
        \mathrm{Var}_{q,y,c}\mleft[\hat{U}(\pi )\mright] =  \frac{1}{N}   \mathbb{E}_{q} \mleft[ \mathrm{Var}_{y,c}\mleft[ \hat{U}(\pi) \mid q\mright] \mright]. 
        \label{eq:vardecom}
    \end{equation}
    Since we have assumed a rank-based examination model (Section~\ref{sec:background:cltr}), the examinations of documents are independent.
    This allows us to rewrite the variance conditioned on a single query:
    \begin{align}
        & \mathrm{Var}_{y,c}\mleft[\hat{U}(\pi \mid q)\mright] =  \mathrm{Var}_{y,c}\mleft[ \sum_{d \in D_{q}}^{}   \frac{\rho(d)}{\rho_{0}(d)} c(d, q) \mright]   \\ 
        &\;\; = \sum_{d \in D_{q}}^{} \mathrm{Var}_{y,c}\mleft[  \frac{\rho(d)}{\rho_{0}(d)} c(d, q)  \mright] 
        \leq  \sum_{d \in D_{q}} \mathbb{E}_{c, y} \mleft[ \mleft( \frac{\rho(d)}{\rho_{0}(d)} c(d, q) \mright)^{2} \mright]. \nonumber
    \end{align}
   Since: $c(d, q)^2 = c(d, q)$, we can further rewrite to:
    \begin{align}
    \sum_{d \in D_{q}} \mathbb{E}_{c, y} \mleft[ \mleft( \frac{\rho(d)}{\rho_{0}(d)} c(d, q) \mright)^{2} \mright]
         &= \sum_{d \in D_q} \mathbb{E}_{c, y} \mleft[ \mleft( \frac{\rho(d)}{\rho_{0}(d)} \mright)^{2} c(d, q)  \mright]   \\
         &=  \sum_{d \in D_q}   \mleft( \frac{\rho(d)}{\rho_{0}(d)} \mright)^{2} P(C=1 \mid d, q, \pi_{0}).    \nonumber
    \end{align}
    Next, we use Eq.~\ref{click-model} and~\ref{policy-aware-exposure} to substitute the click probability;
    subsequently, we replace the examination propensities with normalized counterparts using Eq.~\ref{norm-exposure} and~\ref{eq:ratio};
    and lastly, we upper bound the result using the fact that $P(R=1 | d, q)  \leq 1$:
\begin{align}
%    \begin{split}
    &\sum_{d \in D_{q}}  \mathbb{E}_{c, y} \mleft[ \mleft( \frac{\rho(d)}{\rho_{0}(d)} c(d, q) \mright)^{2} \mright] \nonumber \\
     &=  \sum_{d \in D_q}    \rho_{0}(d) \mleft( \frac{\rho(d)}{\rho_{0}(d)} \mright)^{2}  P(R=1 \,|\, d, q)
    \\[-1ex]
    &= \sum_{d \in D_q}  \mathrm{Z}\, \rho'_{0}(d) \mleft( \frac{\rho'(d)}{\rho'_{0}(d)} \mright)^{2}  P(R=1 | d, q)   
       \leq \mathrm{Z}   \sum_{d \in D_q}  \rho'_{0}(d) \mleft( \frac{\rho'(d)}{\rho'_{0}(d)} \mright)^{2}. \nonumber
%   \end{split}
\end{align}
   Finally, we place this upper bound for a single query back into the expectation over all queries (Eq.~\ref{variance-full}):
    \begin{equation}
        \frac{1}{N} \,  \mathbb{E}_{q} \mleft[ \mathrm{Var}_{y,c} \mleft[ \hat{U}(\pi ) \mid q \mright] \mright]   
         \leq \frac{ \mathrm{Z} }{N} \mathbb{E}_{q} \bigg[ \sum_{d \in D_q}  \rho'_{0}(d) \mleft( \frac{\rho'(d)}{\rho'_{0}(d)} \mright)^{2}  \bigg].
         \label{eq:varprooflast}
    \end{equation}
    Therefore, by Eq.~\ref{eq:vardecom}, \ref{eq:varprooflast} and the definition of exposure-based divergence in Eq.~\ref{eq:exposuredivergence}, it is a proven upper bound of the variance.
\end{proof}

\subsection{Exposure-divergence bound on performance}
\label{sec:perfbound}

Using the upper bound on the variance of an \ac{CLTR} \ac{IPS} estimator that was proven in Theorem~\ref{var-theorm}, we can now introduce a generalization bound for the \ac{CLTR} estimator.

\begin{theorem}
\label{CLTR-bound}
    Given the true utility $U(\pi)$ (Eq.~\ref{true-utility}) and its exposure-based \ac{IPS} estimate $\hat{U}(\pi)$ (Eq.~\ref{cltr-obj}), for the ranking policy $\pi$ and the logging policy $\pi_{0}$ with expected exposures $\rho(d)$ and $\rho_{0}(d)$, respectively,
    the following generalization bound holds with probability $1 - \delta$:
\begin{equation}
    U(\pi) \geq \hat{U}(\pi) -  \sqrt{ \frac{\mathrm{Z}}{N}  \Big(\frac{1-\delta}{\delta}\Big) d_2(\rho \,\Vert\, \rho_0)}.
\label{variance}
\end{equation}
\end{theorem}
\begin{proof}
    As per Cantelli's inequality~\cite{ghosh2002probability}, given an estimator $\hat{X}$ with expected value $\mathbb{E}[\hat{X}]$ and variance $\mathrm{Var}[\hat{X}]$, the following tail-bound holds:
    \begin{equation}
        P(\hat{X} - \mathbb{E}[\hat{X}] \geq \lambda ) \leq \frac{\mathrm{Var}[\hat{X}]}{\mathrm{Var}[\hat{X}] + \lambda^2}. 
    \label{cantelli}
    \end{equation} 
    Since $\lambda > 0$ is a free parameter, we can define $\delta$ such that:
    \begin{equation}
        \delta \coloneq  \frac{\mathrm{Var}[\hat{X}]}{\mathrm{Var}[\hat{X}] + \lambda^2},
         \qquad
         \lambda = \sqrt{\frac{1-\delta}{\delta} \mathrm{Var}[\hat{X}]}.
    \end{equation} 
    Consequently, the following inequality holds:
    \begin{equation}
        P(\mathbb{E}[\hat{X}] \geq \hat{X} - \lambda ) \geq 1-\delta.
         \label{cantelli-inequality} 
    \end{equation}
    Building on this inequality, the following inequality must hold with probability $1-\delta$:
    \begin{equation}
        U(\pi) \geq \hat{U}(\pi) - \sqrt{ \frac{1-\delta}{\delta} \mathrm{Var}_{q,y,c}\mleft[\hat{U}(\pi)\mright]}.
        \label{inequality} 
    \end{equation}
    Finally, we can replace the variance with the upper bound from Theorem~\ref{var-theorm}, which completes the proof.
\end{proof}

\begin{figure}[t]
\centering
\includegraphics[width=0.95\columnwidth,trim= 0 -0.25cm 0 0,clip]{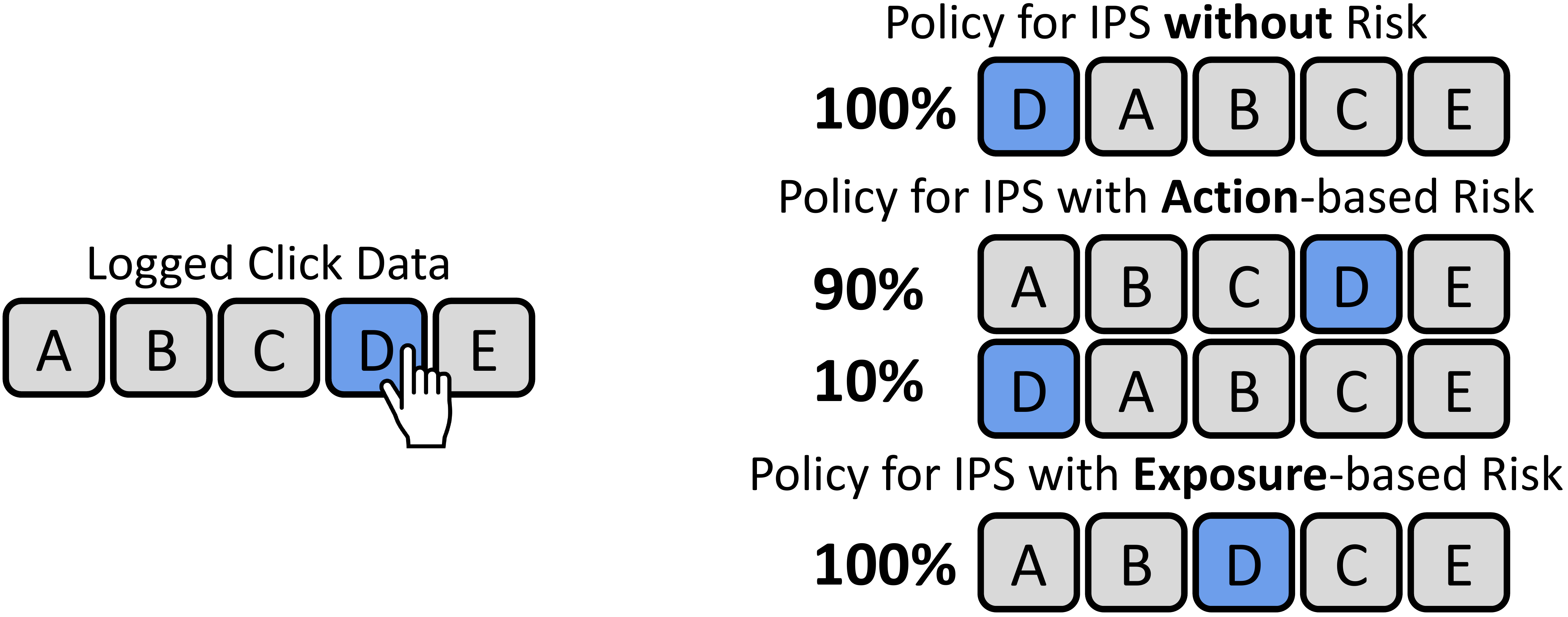}
\caption{
Example comparison of the optimal policy for a single logged click according to three different risk estimators.
}
\label{fig:riskcomparison}
\end{figure}

\noindent\textbf{Risk in \ac{CLTR}.}
Based on the generalization bound proposed in Theorem~\ref{CLTR-bound}, we see that it proposes the following measure of risk:
$\text{Risk}(\delta) = \sqrt{ \frac{\mathrm{Z}}{N}  \big(\frac{1-\delta}{\delta}\big) d_2(\rho \,\Vert\, \rho_0)}$ (cf.\ Eq.~\ref{eq:generalizationbound}).
Clearly, this risk is mostly determined by the exposure-based divergence between the new policy and the logging policy.
Thereby, it states that the greater the difference between how exposure is spread over documents by the logging policy and the new policy, the higher the risk involved.
Therefore, to optimize this lower bound, one has to balance the maximization of the estimated utility $\hat{U}(\pi)$ and the minimization of risk by not letting $\pi$ differ too much from $\pi_0$ in terms of exposure.

Furthermore, we see that our measure of risk diminishes as $N$ increases.
As a result, the risk term will overwhelm the \ac{IPS} term when $N$ is very low, as there is much risk involved when estimating based on a few interactions.
Conversely, when $N$ is very large, the risk term mostly disappears, as the \ac{IPS} estimate is more reliable when based on large numbers of interactions.
Thus, during optimization, the generalization bound is expected to mostly help with avoiding initial decreases in performance, while still converging at the same place as the standard \ac{IPS} estimator.

Lastly, the $\delta$ parameter determines the \emph{safety} that is provided by the risk, where a lower $\delta$ makes it more likely that the generalization bound holds.
Accordingly, as $\delta$ increases the risk term becomes smaller and will thus have less effect on optimization.

To the best of our knowledge, this is the first exposure-based generalization bound, which makes it the first method designed for safe optimization in the \ac{CLTR} setting.

\header{Illustrative comparison}
To emphasize the working and novelty of our exposure-based risk, a comparison of the optimal policies for action-based risk, exposure-based risk, and no risk are shown in Figure~\ref{fig:riskcomparison}.
We see that \ac{IPS} without a risk term places the once-clicked document at the first position, with 100\% probability.
This is very risky, as it greatly impacts the ranking while only being based on a single observation.
The action-based risk tries to mitigate this risk with a probabilistic policy that gives most probability to the logging policy ranking (90\%) and the remainder to the IPS ranking (10\%).
In contrast, with exposure-based risk, the optimal policy makes the risk and utility trade-off in a single ranking, that mostly follows the logging policy but places the clicked document slightly higher.

This example illustrates that because action-based risk does not have a similarity measure between rankings, it can only produce a probabilistic interpolation between the logging policy and IPS rankings.
Alternatively, because exposure-based risk does have such a measure, it produces a ranking that is neither the logging ranking nor the IPS ranking, but one with an exposure distribution that is similar to both.
Thereby, exposure-based risk has a more elegant and natural method of balancing utility maximization and risk minimization in the \ac{CLTR} setting.

\begin{table*}[ht]
\setlength{\tabcolsep}{0.02cm}
\centering
\caption{
NDCG@5 performance under different settings and datasets for several values of $N$, the number of logged interactions in the simulated training set.
Reported numbers are averages over 10 independent runs evaluated on the held-out test-sets, bold numbers indicate the highest performance.
Statistical significance for differences with the exposure-based \ac{CRM} are measured via a two-sided student-t test,
$^{\tiny \blacktriangledown}$ indicates methods with significantly lower NDCG with $p<0.01$, and $^{\tiny \star}$ no significant difference.
}
\label{tab:dcgresults}
\resizebox{\textwidth}{!}{
\input{sections/table.tex}
    }
\end{table*}

\section{A Novel Counterfactual Risk Minimization Method for LTR}
\label{risk-cltr}

Now that we have the proven generalization bound described in Section~\ref{sec:perfbound} (Theorem~\ref{CLTR-bound}),
we can propose a novel risk-aware \ac{CLTR} method for optimizing it.
The aim of our method is to find the policy that maximizes this high-confidence lower bound on the true performance.
In formal terms, we have the following optimization problem:
\begin{equation}
    \max_{\pi}  \hat{U}(\pi) -  \sqrt{ \frac{\mathrm{Z}}{N}  \Big(\frac{1-\delta}{\delta}\Big) d_2(\rho \,\Vert\, \rho_0)}.
\label{objgenbound}
\end{equation}
We propose to train a stochastic policy $\pi$ via stochastic gradient descent, therefore, we need to derive the gradient and find a method of computing it.
For the computation of the gradient w.r.t.\ the utility $\hat{U}(\pi)$, the first part of Eq.~\ref{objgenbound}, we refer to several prior work that discusses this topic extensively~\citep{oosterhuis2021computationally, oosterhuis2020policy, yadav2021policy}.
Thus, we can focus our attention on the second part of Eq.~\ref{objgenbound}:
\begin{equation}
\nabla_{\!\pi} \sqrt{ \frac{\mathrm{Z}}{N}  \Big(\frac{1-\delta}{\delta}\Big) d_2(\rho \,\Vert\, \rho_0)}
         = \sqrt{ \frac{\mathrm{Z}(1-\delta)}{ 4N\delta d_2(\rho \,\Vert\, \rho_0)}} \nabla_{\!\pi} d_2(\rho \,\Vert\, \rho_0).
\end{equation}
To derive the gradient of the exposure-based divergence function, we use the relation between $\rho$ and $\rho'$ from Eq.~\ref{total-exposure} and \ref{eq:ratio}:
\begin{equation}
\begin{split}
    \nabla_{\!\pi} d_2(\rho \,\Vert\, \rho_0)
    &=  \nabla_{\!\pi} \mathbb{E}_{q} \bigg[ \sum_{d \in D_q} \!  \rho'_{0}(d) \mleft( \frac{\rho'(d)}{\rho'_{0}(d)} \mright)^{2}   \bigg] 
    \\
    &= \frac{2}{\mathrm{Z}} \, \mathbb{E}_{q} \bigg[ \! \sum_{d \in D_q} \!\! \frac{\rho(d)}{\rho_{0}(d)} \nabla_{\!\pi} \rho(d)   \bigg] .
\end{split}
\end{equation}
Thus, we only need the gradient w.r.t.\ the exposure of a document ($\nabla_{\!\pi} \rho(d)$) to complete our derivation.
If $\pi$ is a \ac{PL} ranking model, one can make use of the specialized gradient computation algorithm from~\citep{oosterhuis2021computationally}.
However, for this work, we will not make further assumptions about $\pi$ and apply the more general log-derivate trick from the REINFORCE algorithm~\cite{williams1992simple}:
\begin{equation}
\nabla_{\!\pi} \rho(d) = \mathbb{E}_{y \sim \pi} \big[ P\big(E=1 \mid \text{rank}(d \mid y)\big)  \big]  \nabla_{\!\pi} \log \pi(y).
\end{equation}
Putting all of the previous elements back together, gives us the gradient w.r.t.\ the exposure-based risk function:
\begin{equation}
\mbox{}\hspace*{-2.5mm}
  \sqrt{\!\frac{
    1-\delta
   }{ N \delta \, \mathrm{Z}  \, d_2(\rho \,\Vert\, \rho_0) }} 
      \mathbb{E}_{q, y \sim \pi} \!\mleft[
      \!\mleft( \sum_{k=1}^K \! \frac{\rho(y_k)}{\rho_{0}(y_k)} P(E=1 |\, k)\!\mright) \nabla_{\!\pi}\!\log \pi(y) 
     \!\mright],
     \hspace*{-2.5mm}\mbox{}
\end{equation}
where $y_k$ is the document at rank $k$ in ranking $y$.
For a close approximation of this gradient, we substitute the gradient with the queries from the given dataset, and the rankings sampled from $\pi$ during optimization~\cite{williams1992simple,oosterhuis2021computationally}.

Similarly, since the exact computation of  is $d_2(\rho \,\Vert\, \rho_0)$ infeasible in practice, we introduce a sample-based empirical divergence estimator:
\begin{equation}
        \hat{d}_{2}(\rho \mid\mid \rho_{0}) = \frac{1}{N} \sum_{i=1}^{N} \sum_{d \in D_{q_i}}^{}   \rho_{0}'(d) \mleft( \frac{\rho'(d)}{\rho_{0}'(d)} \mright)^{2} .
    \label{empirical-div}
\end{equation}
This is an unbiased estimate of the true divergence given that the sampling process is truly monte-carlo~\cite{james1980monte}.

%% file: sections/table.tex
\begin{tabular}{ l ccc  ccc ccc}
 \toprule
&\multicolumn{3}{c}{ Yahoo! Webscope}
&\multicolumn{3}{c}{ MSLR-WEB30k}
&\multicolumn{3}{c}{ Istella}\\
\cmidrule(r){2-4}
\cmidrule(r){5-7}
\cmidrule{8-10}
& $N=4\cdot10^2$&  $N=4 \cdot 10^7$ &  $N=10^9$&  $N=4 \cdot 10^2$ &  $N=4 \cdot 10^7$ &  $N=10^9$&  $N=4 \cdot 10^2$ &  $N=4 \cdot 10^7$ &  $N=10^9$\\
 \midrule
\multicolumn{1}{l}{ Logging} &  0.677 \phantom{\small(0.000)}\phantom{$^{\tiny -}$}  &  0.677 \phantom{\small(0.000)}\phantom{$^{\tiny -}$}  &  0.677 \phantom{\small(0.000)}\phantom{$^{\tiny -}$}  &  0.435 \phantom{\small(0.000)}\phantom{$^{\tiny -}$}  &  0.435 \phantom{\small(0.000)}\phantom{$^{\tiny -}$}  &  0.435 \phantom{\small(0.000)}\phantom{$^{\tiny -}$}  &  0.635 \phantom{\small(0.000)}\phantom{$^{\tiny -}$}  &  0.635 \phantom{\small(0.000)}\phantom{$^{\tiny -}$}  &  0.635 \phantom{\small(0.000)}\phantom{$^{\tiny -}$}  \\ 
\multicolumn{1}{l}{ Skyline}&  0.727 \phantom{\small(0.000)}\phantom{$^{\tiny -}$} &  0.727 \phantom{\small(0.000)}\phantom{$^{\tiny -}$} &  0.727 \phantom{\small(0.000)}\phantom{$^{\tiny -}$} &  0.479 \phantom{\small(0.000)}\phantom{$^{\tiny -}$} &  0.479 \phantom{\small(0.000)}\phantom{$^{\tiny -}$} &  0.479 \phantom{\small(0.000)}\phantom{$^{\tiny -}$} &  0.714 \phantom{\small(0.000)}\phantom{$^{\tiny -}$} &  0.714 \phantom{\small(0.000)}\phantom{$^{\tiny -}$} &  0.714 \phantom{\small(0.000)}\phantom{$^{\tiny -}$} \\ \midrule
\multicolumn{1}{l}{ Naive}&  0.652 \small (0.021)$^{\tiny \blacktriangledown}$ &  0.694 \small (0.000)$^{\tiny \blacktriangledown}$ &  0.695 \small (0.000)$^{\tiny \blacktriangledown}$ &  0.353 \small (0.003)$^{\tiny \blacktriangledown}$ &  0.448 \small (0.000)$^{\tiny \blacktriangledown}$ &  0.448 \small (0.001)$^{\tiny \blacktriangledown}$ &  0.583 \small (0.007)$^{\tiny \blacktriangledown}$ &  0.661 \small (0.001)$^{\tiny \blacktriangledown}$ &  0.661 \small (0.001)$^{\tiny \blacktriangledown}$  \\ 
 \multicolumn{1}{l}{Action IPS}&  0.656 \small (0.008)$^{\tiny \blacktriangledown}$ &  0.701 \small (0.001)$^{\tiny \blacktriangledown}$ &  0.701 \small (0.001)$^{\tiny \blacktriangledown}$ &  0.359 \small (0.007)$^{\tiny \blacktriangledown}$ &  0.448 \small (0.001)$^{\tiny \blacktriangledown}$ &  0.448 \small (0.001)$^{\tiny \blacktriangledown}$ &  0.578 \small (0.004)$^{\tiny \blacktriangledown}$ &  0.671 \small (0.001)$^{\tiny \blacktriangledown}$ &  0.671 \small (0.002)$^{\tiny \blacktriangledown}$  \\ 
 \multicolumn{1}{l}{Action \ac{CRM}}&  0.617 \small (0.004)$^{\tiny \blacktriangledown}$ &  0.698 \small (0.001)$^{\tiny \blacktriangledown}$ &  0.700 \small (0.001)$^{\tiny \blacktriangledown}$ &  0.359 \small (0.005)$^{\tiny \blacktriangledown}$ &  0.448 \small (0.001)$^{\tiny \blacktriangledown}$ &  0.449 \small (0.001)$^{\tiny \blacktriangledown}$ &  0.449 \small (0.013)$^{\tiny \blacktriangledown}$ &  0.668 \small (0.002)$^{\tiny \blacktriangledown}$ &  0.672 \small (0.001)$^{\tiny \blacktriangledown}$  \\ 
 \multicolumn{1}{l}{Exp. IPS}&  0.659 \small (0.010)$^{\tiny \blacktriangledown}$ &  \textbf{0.723 \small (0.001)}$^{\tiny \star}$ &  \textbf{0.730 \small (0.001)}$^{\tiny \star}$ &  0.389 \small (0.014)$^{\tiny \blacktriangledown}$ &  \textbf{0.474 \small (0.001)}$^{\tiny \star}$ &  \textbf{0.481 \small (0.001)}$^{\tiny \star}$ &  0.576 \small (0.010)$^{\tiny \blacktriangledown}$ &  \textbf{0.696 \small (0.001)}$^{\tiny \star}$ &  \textbf{0.706 \small (0.001)}$^{\tiny \star}$  \\  
 \multicolumn{1}{l}{Exp. \ac{CRM}}&  \textbf{0.677 \small (0.001)}\phantom{$^{\tiny -}$} &  \textbf{0.723 \small (0.001)}\phantom{$^{\tiny -}$} &  \textbf{0.730 \small (0.000)}\phantom{$^{\tiny -}$} &  \textbf{0.434 \small (0.001)}\phantom{$^{\tiny -}$} &  \textbf{0.473 \small (0.001)}\phantom{$^{\tiny -}$} &  \textbf{0.480 \small (0.001)}\phantom{$^{\tiny -}$} &  \textbf{0.635 \small (0.001)}\phantom{$^{\tiny -}$} &  \textbf{0.695 \small (0.001)}\phantom{$^{\tiny -}$} &  \textbf{0.706 \small (0.001)}\phantom{$^{\tiny -}$}  \\ 
 \bottomrule
\end{tabular}

%% file: sections/05-experiments.tex
{\renewcommand{\arraystretch}{0.01}
\setlength{\tabcolsep}{0.04cm}
\begin{figure*}[th]
\centering
\begin{tabular}{c r r r }
&
 \multicolumn{1}{c}{ \small \hspace{0.5cm} Yahoo! Webscope}
&
 \multicolumn{1}{c}{ \small \hspace{0.5cm} MSLR-WEB30k}
&
 \multicolumn{1}{c}{ \small \hspace{0.5cm} Istella}
\\
\rotatebox[origin=lt]{90}{\hspace{0.77cm}\small  NDCG@5} &
\includegraphics[scale=0.475]{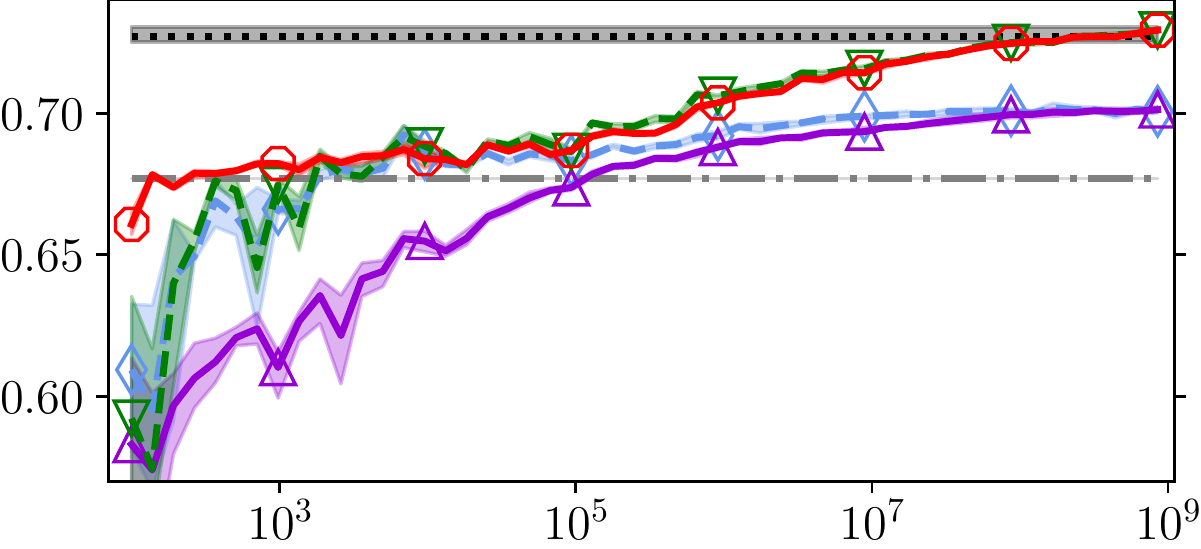} &
\includegraphics[scale=0.475]{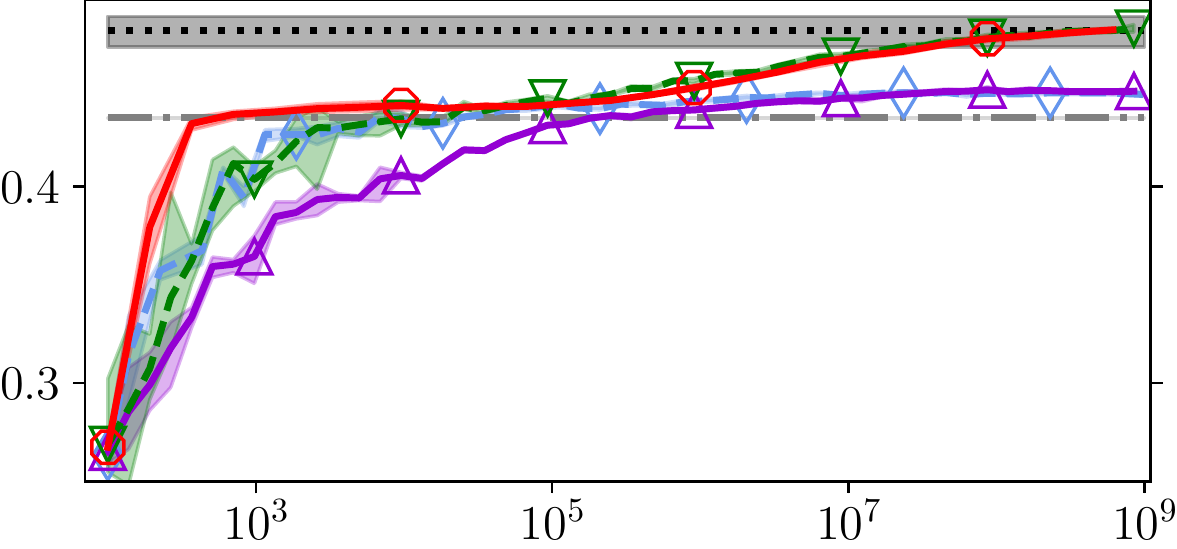} &
\includegraphics[scale=0.475]{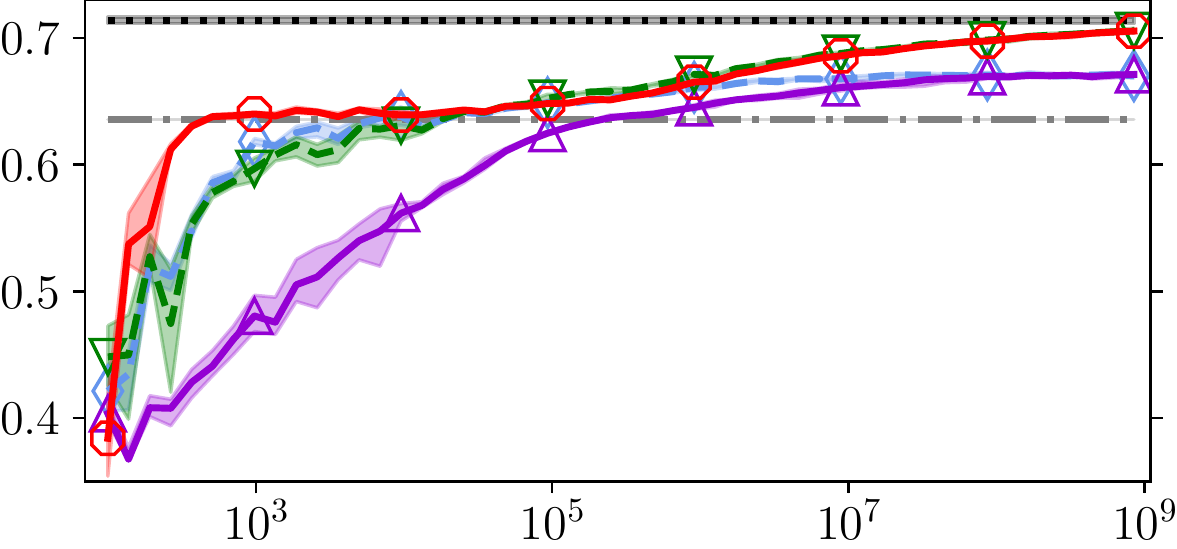}
\\
\rotatebox[origin=lt]{90}{\hspace{0.65cm} \small NDCG@5} &
\includegraphics[scale=0.475]{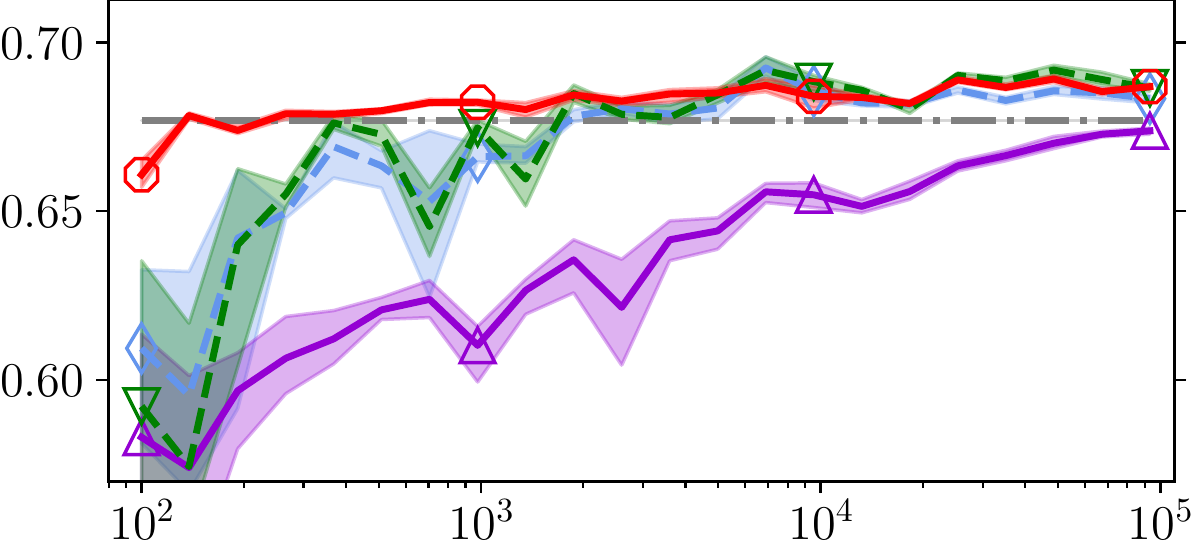} &
\includegraphics[scale=0.475]{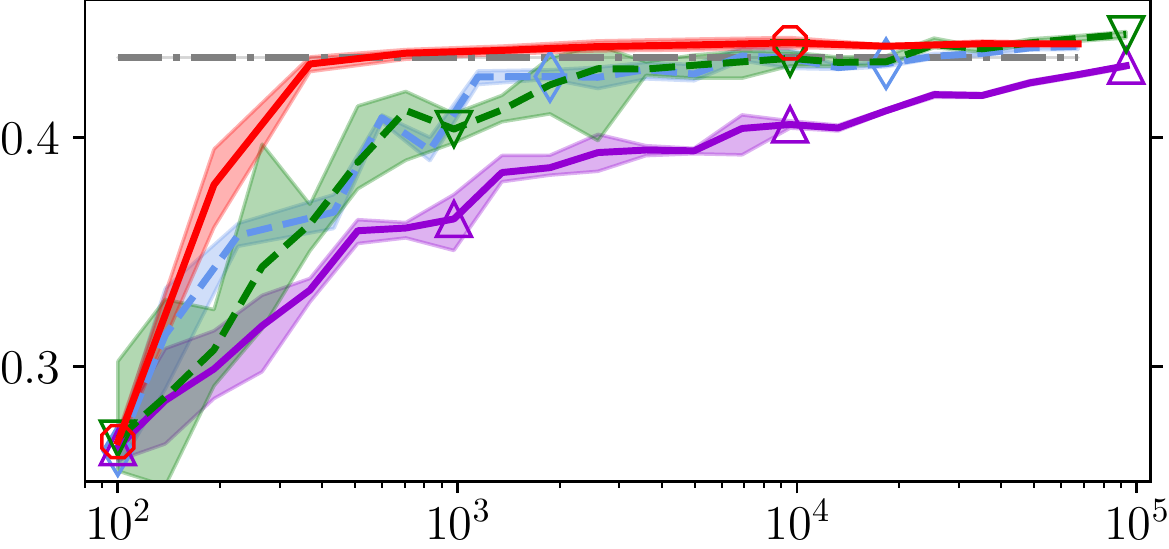} &
\includegraphics[scale=0.475]{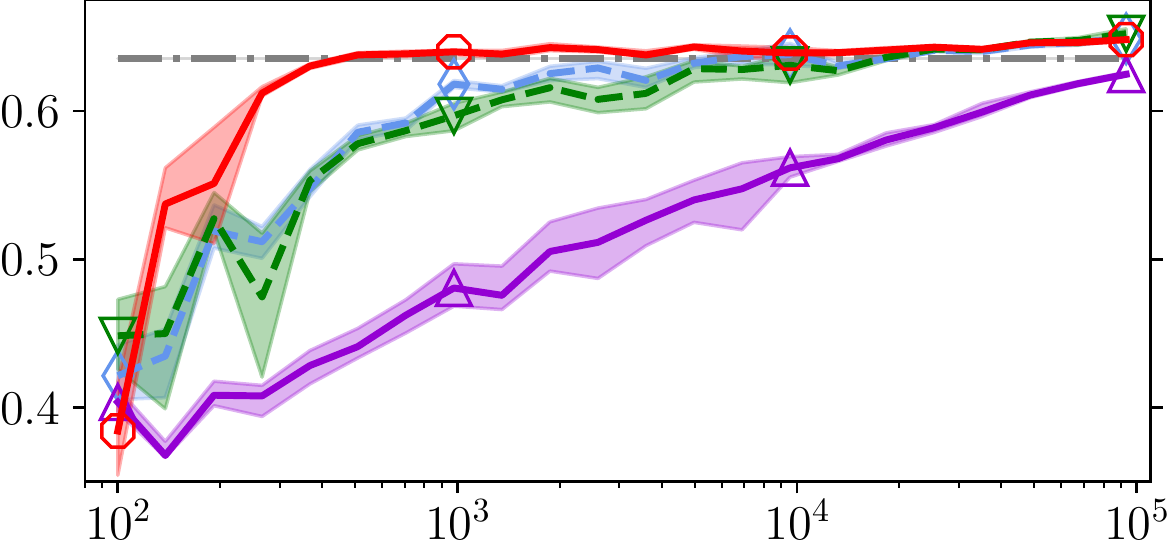} 
\\
& \multicolumn{1}{c}{\small \hspace{1.75em} Number of interactions simulated ($N$)}
& \multicolumn{1}{c}{\small \hspace{1.75em} Number of interactions simulated ($N$)}
& \multicolumn{1}{c}{\small \hspace{1.75em} Number of interactions simulated ($N$)}
\\[2mm]
    \multicolumn{4}{c}{
    \includegraphics[scale=0.47]{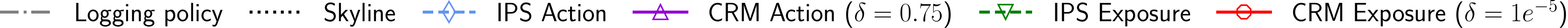}
}
\end{tabular}
\vspace{0.3\baselineskip}
\caption{
Performance in NDCG@5 of various \ac{IPS} and \ac{CRM} methods for \ac{CLTR}.
The top-row presents the results when the size of the training data is varied from extremely small ($10^2$) to extremely high ($10^9$).
The bottom-row is a zoomed-in view, focusing on the low-data region from $10^2$ to $10^5$. 
Results are averages over 10 runs; shaded areas indicate 80\% confidence intervals. 
}
\label{fig:mainresults}
\end{figure*}
}

\section{Experimental Setup}

For our experiments, we follow the semi-synthetic experimental setup that is common in the \ac{CLTR} literature~\citep{oosterhuis2021unifying,oosterhuis2021robust,joachims2017unbiased,vardasbi2020inverse}.
We make use of the three largest publicly available \ac{LTR} datasets: Yahoo!\ Webscope~\cite{chapelle2011yahoo}, MSLR-WEB30k~\citep{qin2013introducing}, and Istella~\citep{dato2016fast}.
The datasets consist of queries, a preselected list of documents per query, query-document feature vectors, and manually-graded relevance judgements for each query-document pair.
To generate clicks, we follow previous work~\citep{oosterhuis2021unifying,oosterhuis2021robust,vardasbi2020inverse} and train a logging policy on a $3\%$ fraction of the relevance judgements.
This simulates a real-world setting, where a production ranker trained on manual judgements is used to collect click logs, which can then be used for subsequent click-based optimization. 
Typically, in real-world ranking settings, given that the production ranker is used on live-traffic, it is deemed as a safe policy that can be trusted with real users.

We simulate a top-$K$ ranking setup~\cite{oosterhuis2020policy} where five documents are presented at once.
Clicks are generated with our assumed click model (Eq.~\ref{click-model}) and the following rank-based position-bias:
\begin{equation}
    P(E=1 \mid q, d, y) = 
\begin{cases}
    \left(\frac{1}{\textrm{rank}(d \vert y)}\right)^2& \text{if } \textrm{rank}(d \mid y) \leq 5,\\
    0              & \text{otherwise}.
\end{cases}    
\end{equation}
In real-world click data, the observed \ac{CTR} is typically very low~\citep{saito2020open,chen2019tiangong,li2010contextual};\ 
hence, to simulate such a sparse click settings, we apply the following transformation from relevance judgements to relevance probabilities:
\begin{equation}
    P(R = 1 \mid q, d) = 0.025 * rel(q,d) + 0.2,
    \label{click-model-simul}
\end{equation}
where $rel(q,d) \in \{0,1,2,3,4\}$ is the relevance judgement for the query-document pair and $0.2$ is added as click noise. 
During training, the only available data consists of clicks generated on the training and validation sets, no baseline method has access to the underlying relevance judgements (expect the skyline).

Furthermore, we assume a setting where the exact logging policy is not available during training.
As a result, the $\hat{\rho}_0$ propensities have to estimated, we use a simple frequency estimate following~\citep{oosterhuis2021unifying}:
\begin{equation}
%\mbox{}\hspace*{-2mm}
    \hat{\rho}_0(d ) = \sum^N_{i=1} \frac{\mathds{1}\big[q = q_i\big]}{\sum^N_{j=1} \mathds{1}\big[q = q_j\big]}  P\big( E= 1 \mid \text{rank}(d \mid y_i)\big). \label{prop-estimate}
%\hspace*{-2mm}\mbox{}    
\end{equation}
For the action-based baselines, the action propensities $\hat{\pi}_0(y \mid q)$ are similarly estimated based on observed frequencies:
\begin{equation}
        \hat{\pi}_0(y \,|\, q) = \prod_{k=1}^{K-1} \hat{\pi}_0(y_k \,|\, q), \hspace*{1mm} %
        \hat{\pi}_0(y_k \,|\, q) = \sum^N_{j=1} \frac{ \mathds{1}\big[y_k = y_j] }{ \sum^N_{j=1} \mathds{1}\big[q = q_j\big] },
    \label{action-estimation}
\end{equation}
where $\hat{\pi}_0(y_k \,|\, q)$ is the estimated probability of $d$ appearing at rank $k$ for query $q$. 
As is common in \ac{CLTR}~\citep{oosterhuis2020learning, saito2021counterfactual, joachims2017unbiased}, we clip propensities by $10 / \sqrt{N}$ in the training set, to reduce variance, but not in the validation set.

We optimize neural \ac{PL} ranking models~\citep{oosterhuis2021computationally} with early stopping based on validation clicks to prevent overfitting. 
For the REINFORCE policy-gradient, we follow~\cite{yadav2021policy} and use the average reward per query as a control-variate for variance reduction.

As our evaluation metric, we compute NDCG@5 metric using the relevance judgements on the test split of each dataset~\citep{jarvelin2002cumulated}. 
All reported results are averages over ten independent runs, significant testing is performed with a two-sided student-t test.

Finally, the following methods are included in our comparisons:
\begin{enumerate}[label=(\roman*), leftmargin=*]
    \item \emph{Naive}. As the most basic baseline, we train on the generated clicks without any correction (equivalent to $\forall d, \, \rho_0(d)=1 $). %
     \item  \emph{Skyline.} To compare with the highest possible performance, this baseline is trained on the actual relevance judgements.
    \item  \emph{Action-based IPS.} Standard IPS estimation (Eq.~\ref{eq:actionips}) that is not designed for ranking and thus uses action-based propensities.
    \item  \emph{Action-based \ac{CRM}.} Standard \ac{CRM} (Eq.~\ref{crm-loss}) that is also not designed for ranking, for the risk function we use the action-based divergence function in Eq.~\ref{eq:actionbaseddiv}.
    \item  \emph{Exposure-based IPS}. The IPS estimator designed for \ac{CLTR} with exposure-based propensities (Eq.~\ref{cltr-obj}).
    The most important baseline, as it is the prevalent approach in the field~\cite{oosterhuis2020policy,oosterhuis2021unifying}.
     \item  \emph{Exposure-based \ac{CRM}.} Our proposed \ac{CRM} method (Eq.~\ref{objgenbound}) using a risk function based on exposure-based divergence.
\end{enumerate}

%% file: sections/06-results.tex
\section{Results and Discussion}

\subsection{Comparison with baseline methods}

The main results of our experimental comparison are presented in Figure~\ref{fig:mainresults} and Table~\ref{fig:mainresults}.
Figure~\ref{fig:mainresults} displays the performance curves of the different methods as the number of logged interactions ($N$) increases.
Table~\ref{fig:mainresults} presents performance at $N\in\{4 \cdot 10^2, 4 \cdot 10^7, 10^9\}$ and indicates whether the observed differences with our exposure-based \ac{CRM} method are statistically significant.

We start by considering the performance curves in Figure~\ref{fig:mainresults}.
We see that both the action-based and exposure-based \ac{IPS} baselines have an initial period of very similar performance that is far below the logging policy.
Around $N\approx10^4$ their performance is comparable to the logging policy, and finally at $N=10^9$ the exposure-based IPS has reached optimal performance, while the performance of action-based IPS is still far from optimal.
We can attribute this initial poor performance to the high variance problem of IPS estimation;
when $N$ is small, variance is at its highest, resulting in risky and sub-optimal optimization by the IPS estimators.
However, even when $N=10^9$, the variance of the action-based IPS estimator is too high to reach optimal performance, due to its extremely small propensities.
This illustrates why the introduction of exposure-based propensities was so important to the \ac{CLTR} field, and that even exposure-based IPS produces unsafe optimization when little data is available or variance from interactions is high.

Next, we consider whether action-based \ac{CRM} is able to mitigate the high variance problem of action-based \ac{IPS}.
Despite being a proven generalization bound, Figure~\ref{fig:mainresults} clearly shows us that action-based \ac{CRM} only leads to decreases in performance compared to its IPS counterpart.
It appears that this happens because the logging policy is not available in our setup, and the propensities have to be estimated from logged data.
Consequently, the action-based risk pushes the optimization to mimic the exact rankings that were observed during logging.
Thus, due to the variance introduced from the sampling of rankings from the logging policy, it appears that action-based \ac{CRM} has an even higher variance problem than action-based \ac{IPS}.
As expected, our results thus clearly indicate that action-based \ac{CRM} is also unsuited for the \ac{CLTR} setting, to our surprise; it is substantially worse than its IPS counterpart.

Finally, we examine the performance of our novel exposure-based \ac{CRM} method.
Similar to the other methods, there is an initial period of low performance, but in stark contrast, this period ends very quickly;
on Yahoo!\ logging policy performance is reached when $N \approx 125$, on MSLR-WEB30k when $N\approx350$ and on Istella when $N\approx400$.
For comparison, exposure-based IPS needs $N\approx1100$ on Yahoo!, $N\approx10^4$ on MLSR-WEB30k and $N\approx1.1\cdot10^4$ on Istella to do the same; meaning that our \ac{CRM} method needs roughly $89\%$, $97\%$ and $97\%$ fewer interactions, respectively.
In addition, Table~\ref{fig:mainresults} indicates that the logging policy performance is matched on all datasets when $N=400$ by exposure-based \ac{CRM}, where it also outperforms all baseline methods.
We note that there is still an initial period of low performance, because the logging policy is unavailable at training, and thus, its behavior still has to be estimated from logged interactions.
It is possible that in settings where the logging policy is fully known during training, this initial period is eliminated entirely.
Nevertheless, our results show that exposure-based \ac{CRM} reduces the initial periods of poor performance due to variance by an enormous magnitude.

{\renewcommand{\arraystretch}{0.01}
\setlength{\tabcolsep}{0.015cm}
\begin{figure*}[t]
\centering
\begin{tabular}{c r r r }
&
 \multicolumn{1}{c}{ \small \hspace{0.5cm} Yahoo! Webscope}
&
 \multicolumn{1}{c}{ \small \hspace{0.5cm} MSLR-WEB30k}
&
 \multicolumn{1}{c}{ \small \hspace{0.5cm} Istella}
\\
\\
\rotatebox[origin=lt]{90}{\hspace{0.2cm} \small \it Action-based \ac{CRM}} 
\rotatebox[origin=lt]{90}{\hspace{0.65cm} \small NDCG@5} &
\includegraphics[scale=0.475]{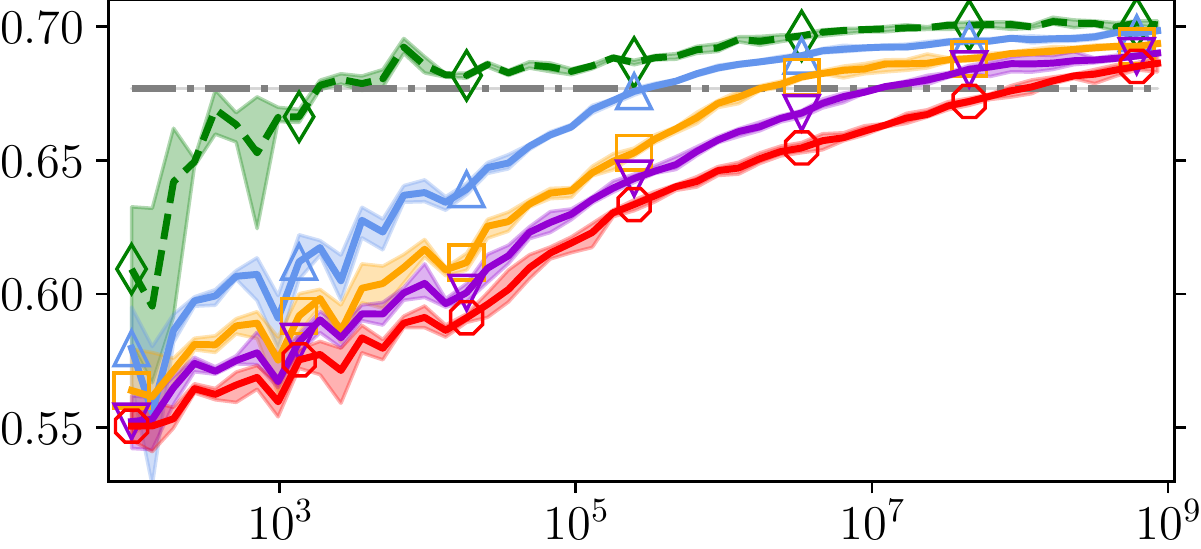} &
\includegraphics[scale=0.475]{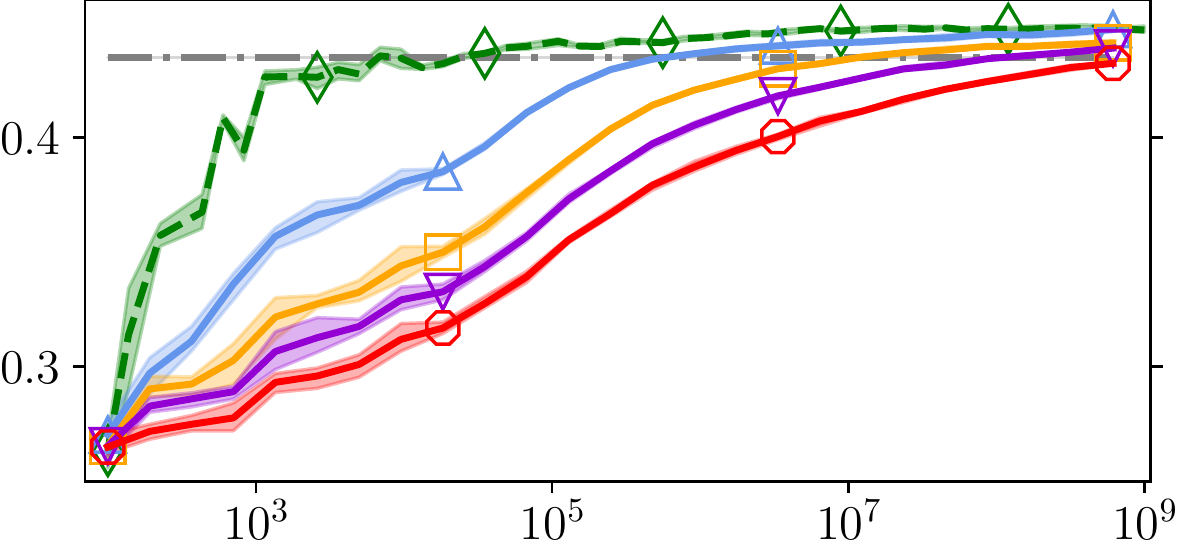} &
\includegraphics[scale=0.475]{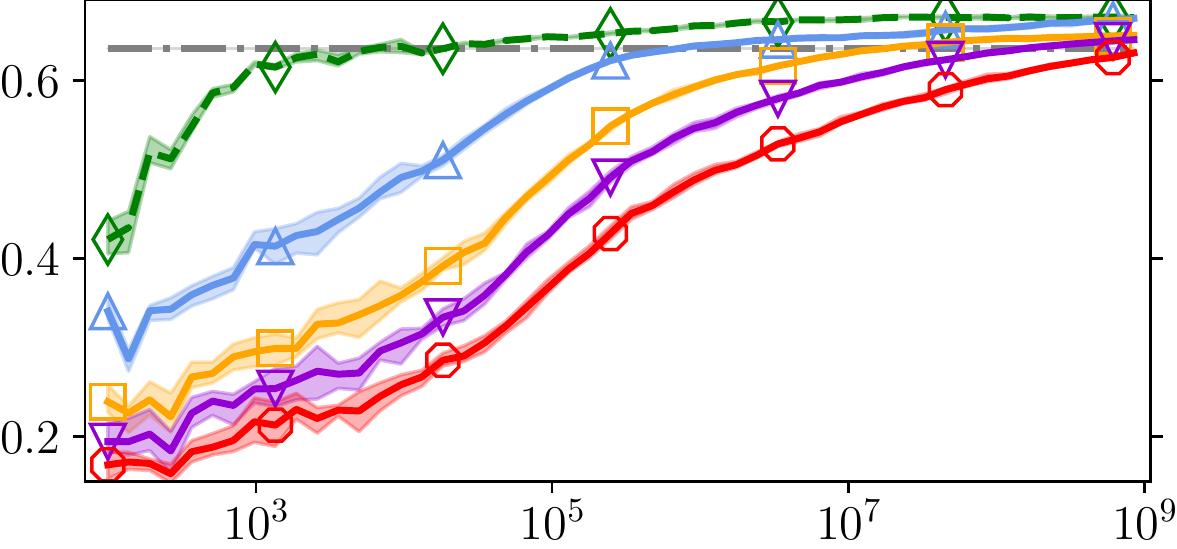} 
\\
\rotatebox[origin=lt]{90}{\hspace{0.1cm}\small\it Exposure-based \ac{CRM}} 
\rotatebox[origin=lt]{90}{\hspace{0.65cm} \small  NDCG@5} &
\includegraphics[scale=0.475]{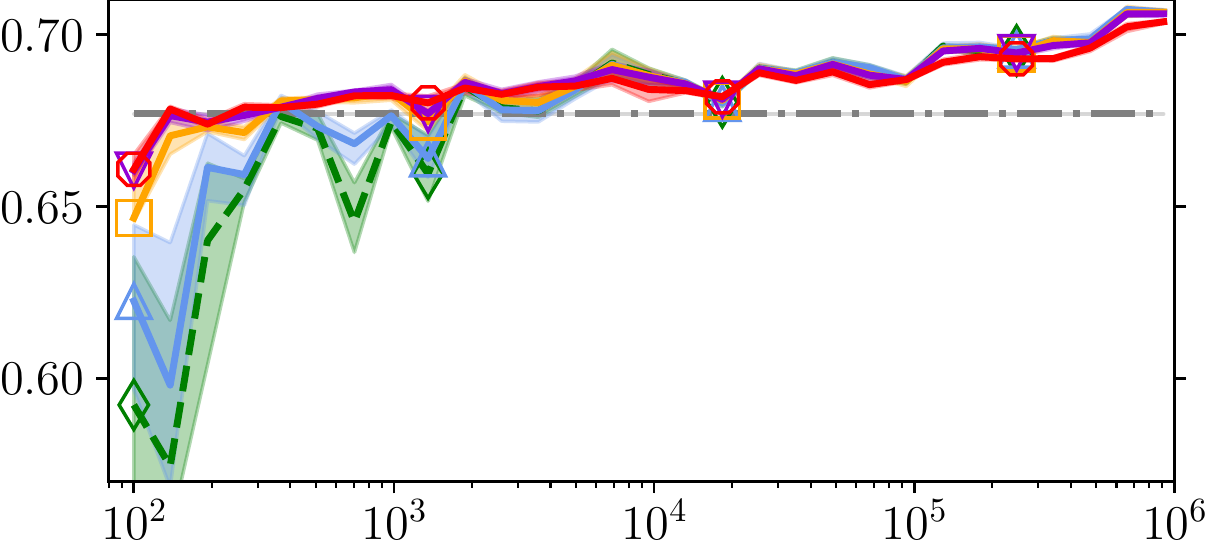} &
\includegraphics[scale=0.475]{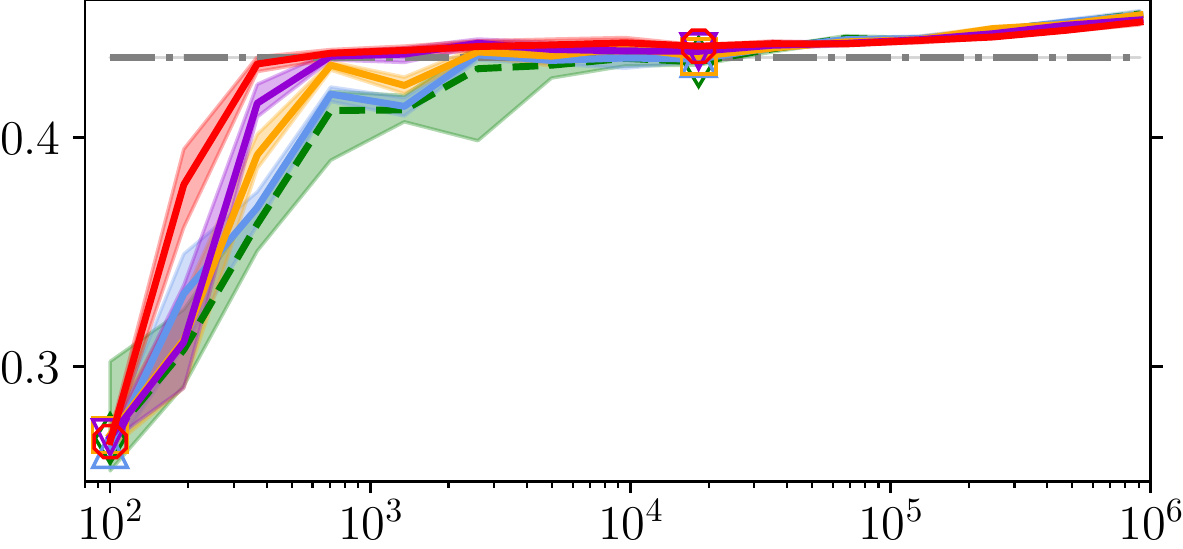} &
\includegraphics[scale=0.475]{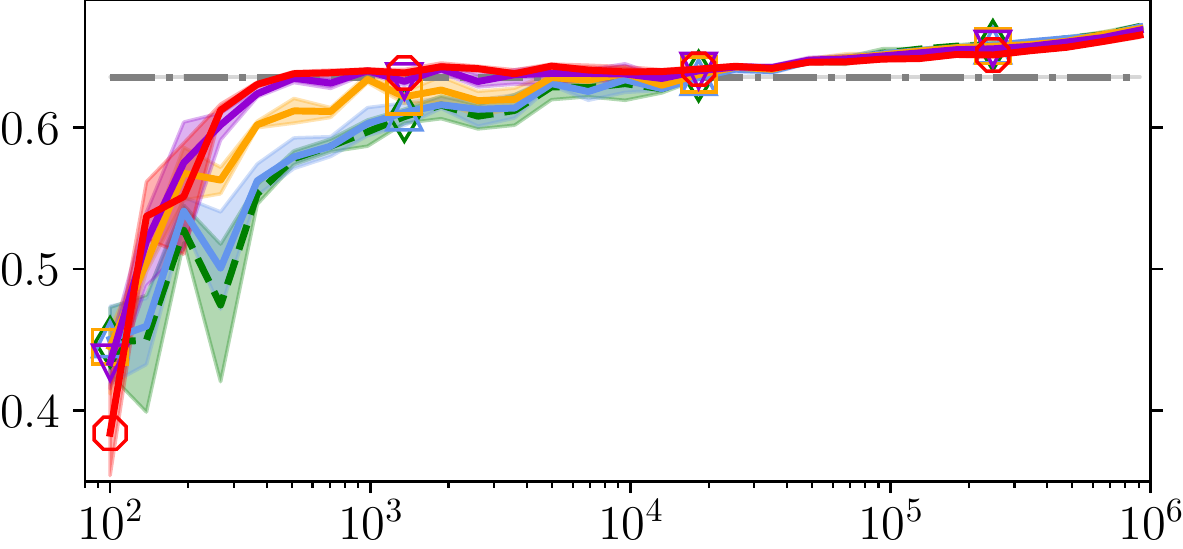}
\\
& \multicolumn{1}{c}{\small \hspace{1.5em} Number of interactions simulated ($N$)}
& \multicolumn{1}{c}{\small \hspace{1.5em} Number of interactions simulated ($N$)}
& \multicolumn{1}{c}{\small \hspace{1.5em} Number of interactions simulated ($N$)}
\\[2mm]
    \multicolumn{4}{c}{
    \includegraphics[scale=0.5]{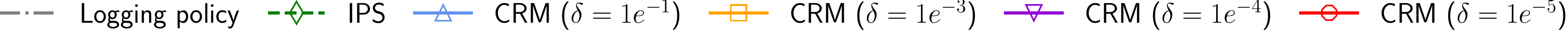}
}
\end{tabular}
\vspace{0.3\baselineskip}
\caption{
    Performance of \ac{CRM} methods with varying confidence parameters ($\delta$).
    Top-row: action-based \ac{CRM} baseline; bottom-row: our exposure-based \ac{CRM} method. 
    Results are averages of 10 runs; shaded areas indicate 80\% confidence intervals. 
}
\label{fig:ablationresults}
\end{figure*}
}

Furthermore, while the initial period is clearly improved, we should also consider whether there is a trade-off with the rate of convergence.
Surprisingly, Figure~\ref{fig:mainresults} does not display any noticeable decrease in performance when compared with exposure-based IPS.
Moreover, Table~\ref{fig:mainresults} shows the differences between exposure-based IPS and \ac{CRM} are barely measurable and not statistically significant when $N \in \{4\cdot 10^7,10^9\}$.
We know from the risk formulation in Eq.~\ref{objgenbound} that the weight of the risk term decreases as $N$ increases at a rate of $1/\sqrt{N}$.
In other words, the more data is available, the more optimization is able to diverge from the logging policy.
It appears that this balances utility maximization and risk minimization so well that we are unable to observe any downside of applying exposure-based \ac{CRM} instead of IPS.
Therefore, we conclude that, compared to all baseline methods and across all datasets, exposure-based \ac{CRM} drastically reduces the initial period of low performance, matches the best rate of convergence of all baseline, and has optimal performance at convergence.

\vspace*{-1mm}
\subsection{Ablation study on the confidence parameter}

To gain insights into how the confidence parameter $\delta$ affects the trade-off between safety and utility, an ablation study over various $\delta$ values was performed for both \ac{CRM} methods. 

The top-row of Figure~\ref{fig:ablationresults} shows us the performance of action-based \ac{CRM}, and contrary to expectation, a decrease in $\delta$ corresponds to a considerably worse performance.
For the sake of clarity, in theory, $\delta$ is inversely tied to safety, a lower $\delta$ should result in less divergence from the safe logging policy~\citep{oosterhuis2021robust}.
Conversely, we see that action-based \ac{CRM} displays the opposite trend.
We think this further confirms our hypothesis that a frequency estimate of action-based divergence has an even higher variance problem than action-based IPS.
Consequently, a higher weight to the risk function results in worse performance.
This further confirms our previous conclusion that action-based \ac{CRM} is unsuited for the \ac{CLTR} setting, regardless of how the $\delta$ parameter is tuned.

In contrast, the bottom-row of Figure~\ref{fig:ablationresults} displays the expected trend for exposure-based \ac{CRM};
as $\delta$ decreases the resulting performance gets closer to the logging policy.
With $\delta=0.1$, \ac{CRM} performs extremely close to its IPS counterpart, as optimization is less constrained to mimic the logging policy here.
Decreasing $\delta$ appears to have diminishing returns, as the difference between $\delta=10^{-4}$ and $\delta=10^{-5}$ is marginal.
Importantly, we do not observe any downsides to setting $\delta=10^{-5}$, thus we have not reached a point in our experiments where $\delta$ is set too conservatively.
This suggests that exposure-based \ac{CRM} is very robust to the setting of the $\delta$ parameter, and that a sufficiently low $\delta$ does not require fine-tuning.
Therefore, this shows that the improvements we observed when comparing with baseline methods, did not stem from a fine-tuning of $\delta$.
Thus, we can conclude that this robustness further increases the safety that is provided by exposure-based \ac{CRM}, as there is also little risk involved in the tuning of the $\delta$ parameter.

%% file: sections/07-conclusion.tex
\section{Conclusion}

In this paper, we introduced the first \acf{CRM} method designed for \ac{CLTR}, that relies on a novel exposure-based divergence function.
In contrast with existing action-based \ac{CRM} methods, exposure-based divergence avoids the problem of the enormous combinatorial action space when ranking, by measuring the dissimilarity between policies based on how they distribute exposure to documents.
As a result, exposure-based \ac{CRM} optimization produces policies that rank similar to the logging policy when it is risky to follow \ac{IPS}, i.e., when little data is available or variance is very high.
Consequently, our experimental results show that it almost completely removes initial periods of detrimental performance;
to be precise, our method needed 89\% to 97\% fewer interactions than state-of-the-art \ac{IPS} to match production system performance.
Importantly, we observed no downsides in its application, as it maintained the same rate and point of convergence as \ac{IPS}, in all tested experimental settings.
Therefore, we conclude that our exposure-based \ac{CRM} method provides the safest \ac{CLTR} methods so far, as it almost completely alleviates the risk of decreasing the performance of a production system.

These improvements have large implications for practitioners who work on ranking systems in real-world settings, since the almost complete reduction of initial detrimental performance removes the main risks involved in applying \ac{CLTR}.
In other words, when applying our novel exposure-based \ac{CRM}, practitioners can have significantly less worry that the resulting policy will perform worse than their production system and hurt user experience.

We hope future work will further research the promising potential applications of exposure-based \ac{CRM}, for instance, in settings with fast turn-around times in deployment, or large numbers of tail-queries~\cite{wedig2006large,white2007studying}, where interaction data is limited.